\documentclass[aps,12pt,twoside,notitlepage,superscriptaddress]{revtex4-1}

\usepackage{graphicx,epic,eepic,epsfig,amsmath,amsfonts,amsthm,latexsym,amssymb,color,xcolor,enumerate}
\usepackage{bm,bbm}
\allowdisplaybreaks[3]
\usepackage[colorlinks,
            linkcolor=blue,
            anchorcolor=blue,
            citecolor=red,
            ]{hyperref}

\newtheorem{theorem}{Theorem}

\newtheorem{lemma}[theorem]{Lemma}

\theoremstyle{definition}
\newtheorem{definition}[theorem]{Definition}
\theoremstyle{remark}
\newtheorem{remark}{Remark}

\newcommand{\ket}[1]{|#1\rangle}
\newcommand{\bra}[1]{\langle#1|}

\newcommand{\proj}[1]{| #1\rangle\!\langle #1 |}

\newcommand{\tr}{\operatorname{Tr}}
\newcommand{\ox}{\otimes}

\newcommand{\supp}{{\operatorname{supp}}}

\newcommand{\rar}{\rightarrow}

\newcommand{\1}{{\openone}}
\newcommand{\id}{{\operatorname{id}}}
\newcommand{\nb}{\nonumber}
\newcommand{\cM}{\mathcal{M}}
\newcommand{\cN}{\mathcal{N}}
\newcommand{\cS}{\mathcal{S}}
\newcommand{\cO}{\mathcal{O}}
\newcommand{\cb}{\mathrm{CB}}

\begin{document}

\title{Completely Bounded Qusi-Norms, Their Mutiplicativity, and New Additivity Results of Quantum Channels}

\author{Ke Li}
    \email{carl.ke.lee@gmail.com}
    \affiliation{Institute for Advanced Study in Mathematics, Harbin Institute of Technology, Harbin 150001, China}	
\author{Quanhua Xu}
    \email{qxu@univ-fcomte.fr}
    \affiliation{Institute for Advanced Study in Mathematics, Harbin Institute of Technology,  Harbin 150001, China; and
	Laboratoire de Math{\'e}matiques, Universit{\'e} Marie \& Louis Pasteur, 25030 Besan\c{c}on Cedex, France}	
    

\begin{abstract}
We obtain two new additivity results of quantum channels. The first one is the additivity of the channel R\'enyi information associated with the sandwiched R\'enyi divergence of order $\alpha\in[\frac{1}{2},1)$. To prove this, we introduce the completely bounded $1\to\alpha$ quasi-norms for completely positive maps, with $\alpha\in[\frac{1}{2},1)$, and show that it is multiplicative. The additivity/multiplicativity derived here extends and complements the results of Devetak {\it et al} (Commun Math Phys 266:37-63, 2006) and Gupta and Wilde (Commun Math Phys 334:867-887, 2015), which deal with the case $\alpha>1$. The second one is the additivity of the channel dispersion, which is a quantity related to the second-order behavior of quantum information tasks.
\end{abstract}

\maketitle
\section{Introduction}
The additivity issue has impacted significantly the development of quantum information theory. Nonadditivity of many entropic formulas is the main obstruction preventing a complete understanding of important quantum problems, such as entanglement~\cite{HHHH2009quantum} and channel capacities~\cite{Hastings2009superadditivity}. On the other hand,  some rare additivity results let us obtain full answers.

In this paper, we present two new additivity results of quantum channels. The first one is the additivity of the channel R\'enyi information associated with the sandwiched R\'enyi divergence of order $\alpha\in[\frac{1}{2},1)$. In the derivation, we introduce the completely bounded $1\to\alpha$ quasi-norms for completely positive maps, with $\alpha\in[\frac{1}{2},1)$,  and show that it is multiplicative. The additivity/multiplicativity obtained in the present paper extends and complements the results of Devetak {\it et al}~\cite{DJKR2006multiplicativity} and Gupta and Wilde~\cite{GuptaWilde2015multiplicativity}, which deal with the case $\alpha>1$. The second one is the additivity of the channel dispersion, which is a quantity related to the second-order behavior of quantum information tasks.


\medskip
\emph{Notations.} Quantum systems are labeled by letters $A$ and $B$, which also represent the associated Hilbert spaces. 
The set of linear operators on $A$ is denoted by $\mathcal{L}(A)$. $\cS(A)$ is the set of quantum states (density operators) on $A$, $\cS_1(A)$ is the set of pure quantum states (density operators with rank one), and $\cS^+(A)$ is the set of quantum states with full rank. We always use $\Phi$ to denote an unnormalized maximally entangled state. That is, $\Phi_{AA'}=\sum_{i,j}\ket{i}\bra{j}_A\ox\ket{i}\bra{j}_{A'}$, where $\{\ket{i}\}_i$ is a set of orthonormal basis of $A$ or $A'$. A completely positive map $\cM_{A\rar B}$ from $\mathcal{L}(A)$ to $\mathcal{L}(B)$ has a Stinespring representation $U: A \rar B\ox E$, such that $\cM(\rho) = \tr_E U\rho U^*$. The map $\cM^C$, which takes $\rho_A$ to $\cM^C(\rho) = \tr_B U\rho U^*$, is called the complementary map of $\cM_{A\rar B}$. $\cM_{A\rar B}$ is a quantum channel if it is further trace-preserving. In this case, $U$ is an isometry and $\cM^C$ is a channel too.

Throughout this paper, all the optimizations, if not specified, are over quantum states on the underlying systems.

\section{Multiplicativity of Completely Bounded Quasi-norms}
For $\alpha\geq 1$ and a completely positive map $\cM:\mathcal{L}(A')\rar\mathcal{L}(B)$, the completely bounded $1\rar\alpha$ norm of $\cM$ has the following simple form~\cite[Theorem 10]{DJKR2006multiplicativity}:
\begin{equation}\label{eq:cbnorm}
\|\cM\|_{\cb,1\rar\alpha}
=\max_{\varphi_{AA'}\in\cS_1(AA')}\frac{\|\id_A\ox\cM(\varphi_{AA'})\|_\alpha}{\|\tr_A\varphi_{AA'}\|_\alpha}.
\end{equation}
In analogy to \eqref{eq:cbnorm}, we extend the definition of the completely bounded $1\rar\alpha$ norm to the range $\alpha\in(0,1)$, as
\begin{equation}\label{eq:cbquasinorm}
\|\cM\|_{\cb,1\rar\alpha}
:=\min_{\varphi_{AA'}\in\cS_1(AA')}\frac{\|\id_A\ox\cM(\varphi_{AA'})\|_\alpha}{\|\tr_A\varphi_{AA'}\|_\alpha},
\quad 0<\alpha<1.
\end{equation}
Let $\cM^C$ be the complementary map of $\cM$, and $\Phi_{AA'}$ be the unnormalized maximally entangled state. The following two expressions for $\|\cM\|_{\cb,1\rar\alpha}$ is equivalent to \eqref{eq:cbquasinorm}, which can be easily verified and will be useful for later derivation.
\begin{equation}\label{eq:cbquasinorma}
\|\cM\|_{\cb,1\rar\alpha}
=\min_{\rho_A\in\cS(A)}\left\|\rho_A^\frac{1}{2\alpha}\cM(\Phi_{AA'})\rho_A^\frac{1}{2\alpha}\right\|_\alpha
=\min_{X_{A'}\geq 0}\frac{\left\|\cM^C(X_{A'})\right\|_\alpha}{\|X_{A'}\|_\alpha}, \quad 0<\alpha<1.
\end{equation}

Our main result is that the completely bounded $1\rar\alpha$ quasi-norms are multiplicative for $\alpha\in[\frac{1}{2}, 1)$.
\begin{theorem}\label{thm:multiplicativity}
Let $\cM_1:\mathcal{L}(A_1')\rar\mathcal{L}(B_1)$ and $\cM_2:\mathcal{L}(A_2')\rar\mathcal{L}(B_2)$ be completely positive maps. For any $\alpha\in[\frac{1}{2}, 1)$, we have
\begin{equation}
\|\cM_1\ox\cM_2\|_{\cb,1\rar\alpha}=\|\cM_1\|_{\cb,1\rar\alpha}\,\|\cM_2\|_{\cb,1\rar\alpha}.
\end{equation}
\end{theorem}
\begin{proof}
The ``$\leq$'' part is obvious. It suffices to prove the other direction. We employ a method similar to~\cite{Jencova2006a}. Let $\cM_1^C:\mathcal{L}(A_1')\rar\mathcal{L}(E_1)$ and $\cM_2^C:\mathcal{L}(A_2')\rar\mathcal{L}(E_2)$ be the complementary maps of $\cM_1$ and $\cM_2$, respectively. Using the second expression of \eqref{eq:cbquasinorma}, we have
\begin{align}
 &\left\|\cM_1\ox\cM_2\right\|_{\cb,1\rar\alpha} \nb\\
=&\min_{X_{A'_1A'_2}\geq0}\frac{\left\|\cM_1^C\ox\cM_2^C(X_{A'_1A'_2})\right\|_\alpha}{\left\|X_{A'_1A'_2}\right\|_\alpha} \nb\\
=&\min_{X_{A'_1A'_2}\geq0}\frac{\left\|(\cM_1^C\ox\id_{E_2})\circ(\id_{A'_1}\ox\cM_2^C)(X_{A'_1A'_2})\right\|_\alpha}{\left\|\id_{A'_1}\ox\cM_2^C(X_{A'_1A'_2})\right\|_\alpha}\frac{\left\|\id_{A'_1}\ox\cM_2^C(X_{A'_1A'_2})\right\|_\alpha}{\left\|X_{A'_1A'_2}\right\|_\alpha} \nb\\
\geq&\min_{X_{A'_1E_2}\geq0}\frac{\left\|(\cM_1\ox\tr_{E_2})^C(X_{A'_1E_2})\right\|_\alpha}{\left\|X_{A'_1E_2}\right\|_\alpha}
     \min_{X_{A'_1A'_2}\geq0}\frac{\left\|(\tr_{A'_1}\ox\cM_2)^C(X_{A'_1A'_2})\right\|_\alpha}{\left\|X_{A'_1A'_2}\right\|_\alpha} \nb\\
=&\left\|\cM_1\ox\tr_{E_2}\right\|_{\cb,1\rar\alpha}\left\|\tr_{A'_1}\ox\cM_2\right\|_{\cb,1\rar\alpha} \nb\\
=&\left\|\cM_1\right\|_{\cb,1\rar\alpha}\left\|\cM_2\right\|_{\cb,1\rar\alpha},
\end{align}
where in the fourth line we have used the fact that the identity map is the complementary of the trace map, and the last line is by Lemma~\ref{lem:tensortr}.
\end{proof}

\begin{lemma}\label{lem:tensortr}
Let $\cM_{A'_1\rar B_1}$ be a completely positive map, and $\tr_{A'_2}$ be the trace map. Then for $\alpha\in[\frac{1}{2}, 1)$,
\begin{equation}
\left\|\cM\ox\tr\right\|_{\cb,1\rar\alpha} = \left\|\cM\right\|_{\cb,1\rar\alpha}.
\end{equation}
\end{lemma}
\begin{proof}
It is easy to see that $\|\tr\|_{\cb,1\rar\alpha}=1$. So,
\begin{equation}\label{eq:trleq}
     \left\|\cM\ox\tr\right\|_{\cb,1\rar\alpha}
\leq \left\|\cM\right\|_{\cb,1\rar\alpha} \left\|\tr\right\|_{\cb,1\rar\alpha}
  =  \left\|\cM\right\|_{\cb,1\rar\alpha}.
\end{equation}
To show the other direction, We use the first expression of \eqref{eq:cbquasinorma} and write
\begin{align}
  &\left\|\cM\ox\tr\right\|_{\cb,1\rar\alpha}^\alpha \nb\\
= &\min_{\rho_{A_1A_2}}\left\|(\rho_{A_1A_2})^\frac{1}{2\alpha}
   \left[\cM\ox\tr_{A'_2}(\Phi_{A_1A'_1}\ox\Phi_{A_2A'_2})\right]
   (\rho_{A_1A_2})^\frac{1}{2\alpha}\right\|_\alpha^\alpha \nb\\
= &\min_{\rho_{A_1A_2}}\left\|\left[\cM(\Phi_{A_1A'_1})\ox\1_{A_2}\right]^\frac{1}{2}
   (\rho_{A_1A_2})^\frac{1}{\alpha}
   \left[\cM(\Phi_{A_1A'_1})\ox\1_{A_2}\right]^\frac{1}{2}\right\|_\alpha^\alpha \nb\\
= &\min_{\rho_{A_1A_2}}\sum_{k=1}^{|A_2|^2}\frac{1}{|A_2|^2}
   \left\|\left[\cM(\Phi_{A_1A'_1})\right]^\frac{1}{2}
    \left(W_{A_2}^{(k)}\rho_{A_1A_2}W_{A_2}^{(k)*}\right)^\frac{1}{\alpha}
   \left[\cM(\Phi_{A_1A'_1})\right]^\frac{1}{2}\right\|_\alpha^\alpha, \label{eq:trgeq-1}
\end{align}
where $\{W_{A_2}^{(k)}\}_k$ is the set of Heisenberg-Weyl operators on $A_2$, and the last equality is because the Schatten qusi-norms is invariant under multiplication of unitary operators. Now, we can apply the convexity property of Lemma~\ref{lem:convexity} (with $q=1$) to lower bound \eqref{eq:trgeq-1}. This lets us obtain
\begin{align}
     &\left\|\cM\ox\tr\right\|_{\cb,1\rar\alpha} \nb\\
\geq &\min_{\rho_{A_1A_2}}\left\|\left[\cM(\Phi_{A_1A'_1})\right]^\frac{1}{2}
      \left(\sum_{k=1}^{|A_2|^2}\frac{1}{|A_2|^2}W_{A_2}^{(k)}\rho_{A_1A_2}W_{A_2}^{(k)*}\right)
      ^\frac{1}{\alpha}\left[\cM(\Phi_{A_1A'_1})\right]^\frac{1}{2}\right\|_\alpha \nb\\
  =  &\min_{\rho_{A_1}}\left\|\left[\cM(\Phi_{A_1A'_1})\right]^\frac{1}{2}
      \left(\rho_{A_1}\ox\frac{\1_{A_2}}{|A_2|}\right)^\frac{1}{\alpha}
      \left[\cM(\Phi_{A_1A'_1})\right]^\frac{1}{2}\right\|_\alpha \nb\\
  =  &\min_{\rho_{A_1}}\left\|\left[\cM(\Phi_{A_1A'_1})\right]^\frac{1}{2}
      \left(\rho_{A_1}\right)^\frac{1}{\alpha}
      \left[\cM(\Phi_{A_1A'_1})\right]^\frac{1}{2}\right\|_\alpha \nb\\
  =  &\left\|\cM\right\|_{\cb,1\rar\alpha}.
\end{align}
\end{proof}

\section{Additivity of Quantum R\'enyi Information}
\begin{definition}[\cite{MDSFT2013on, WWY2014strong}]
\label{def:sandwiched-RD}
Let $\alpha \in [\frac{1}{2},1)\cup(1,+\infty)$, and let $\rho$ be a quantum state and $\sigma$ be positive semidefinite. When $\alpha > 1$ and $\supp(\rho)\subseteq\supp(\sigma)$ or $\alpha \in [\frac{1}{2},1)$ and $\supp(\rho)\not\perp\supp(\sigma)$, the sandwiched R{\'e}nyi divergence of order $\alpha $ is defined as
\begin{equation}
D_{\alpha}(\rho \| \sigma)
:=\frac{1}{\alpha-1} \log \tr {({\sigma}^{\frac{1-\alpha}{2\alpha}} \rho {\sigma}^{\frac{1-\alpha}{2\alpha}})}^\alpha;
\end{equation}
otherwise, we set $D_{\alpha}(\rho \| \sigma)=+\infty$.
\end{definition}
The quantum R\'enyi information of a channel $\cN_{A'\rar B}$, based on this divergence, is defined as~\cite{GuptaWilde2015multiplicativity}
\begin{equation}\label{eq:renyi-I-def}
I_{\alpha}(\cN)
:=\max_{\rho_{AA'}}\min_{\sigma_B}D_\alpha(\cN(\rho_{AA'})\|\rho_A\ox\sigma_B),
\end{equation}
where the maximization is over all pure states $\rho_{AA'}\in\cS_1(AA')$ with $A\cong A'$, and the minimization is over all states $\sigma_B\in\cS(B)$.

\begin{remark}
The maximization in \eqref{eq:renyi-I-def} can be over all states $\rho_{AA'}\in\cS(AA')$, and over all positive integers for the dimension of $A$. However, due to the data processing inequality of the sandwiched R{\'e}nyi divergence~\cite{FrankLieb2013monotonicity}, this does not make any difference compared to the current definition. The quantity
$\min_{\sigma_B}D_\alpha(\cN(\rho_{AA'})\|\rho_A\ox\sigma_B)$ is a quantum R{\'e}nyi mutual information of the state $\cN(\rho_{AA'})$~\cite{WWY2014strong, Beigi2013sandwiched}. So, $I_{\alpha}(\cN)$ is the maximal R{\'e}nyi mutual information that the channel $\cN$ can generate.
\end{remark}

When $\alpha>1$, the quantum R\'enyi information characterizes the strong converse exponent of channel communication~\cite{GuptaWilde2015multiplicativity, CMW2016strong, LiYao2022strong} as well as the reliability function of channel simulation~\cite{LiYao2021reliable}, for which the additivity of $I_\alpha(\cN)$~\cite{GuptaWilde2015multiplicativity} has played an important role. In this paper, we show that for $\alpha\in[\frac{1}{2}, 1)$, the multiplicativity of the completely bounded quasi-norms of Theorem~\ref{thm:multiplicativity} implies the additivity of the quantum R\'enyi information. This extends  the multiplicativity/additivity results of the case $\alpha>1$ established in~\cite{DJKR2006multiplicativity} and~\cite{GuptaWilde2015multiplicativity}.
\begin{theorem}\label{thm:additivity1}
Let $\cN_1:\mathcal{L}(A_1')\rar\mathcal{L}(B_1)$ and $\cN_2:\mathcal{L}(A_2')\rar\mathcal{L}(B_2)$ be two quantum channels. For any $\alpha\in[\frac{1}{2}, 1)$, it holds that
\begin{equation}
I_\alpha(\cN_1\ox\cN_2)=I_\alpha(\cN_1)+I_\alpha(\cN_2).
\end{equation}
\end{theorem}

\begin{proof}
We start with a reformulation of $I_\alpha(\cN)$. By \eqref{eq:renyi-I-def} and Definition~\ref{def:sandwiched-RD}, we can write for the channel $\cN_{A'\rar B}$ that
\begin{align}
I_{\alpha}(\cN)
=&\max_{\rho_{A}}\min_{\sigma_B}
  D_\alpha\left(\sqrt{\rho_A}\cN(\Phi_{AA'})\sqrt{\rho_A}\|\rho_A\ox\sigma_B\right) \nb \\
=&\frac{1}{\alpha-1}\log \min_{\rho_{A}}\max_{\sigma_B}\tr\left[\left(\rho_A^\frac{1}{2\alpha}
  \ox\sigma_B^\frac{1-\alpha}{2\alpha}\right)\cN(\Phi_{AA'})
  \left(\rho_A^\frac{1}{2\alpha}\ox\sigma_B^\frac{1-\alpha}{2\alpha}\right)\right]^\alpha. \label{eq:additivity-1}
\end{align}
Set
\begin{align}
f(\rho_A, \sigma_B)
:=&\tr\left[\left(\rho_A^\frac{1}{2\alpha}\ox\sigma_B^\frac{1-\alpha}{2\alpha}\right)\cN(\Phi_{AA'})
    \left(\rho_A^\frac{1}{2\alpha}\ox\sigma_B^\frac{1-\alpha}{2\alpha}\right)\right]^\alpha \nb \\
 =&\tr\left[\cN(\Phi_{AA'})^\frac{1}{2}\left(\rho_A^\frac{1}{\alpha}\ox\sigma_B^\frac{1-\alpha}{\alpha}\right)
  \cN(\Phi_{AA'})^\frac{1}{2}\right]^\alpha.
\end{align}
Then we have: (a) the function $\rho_A\mapsto f(\rho_A, \sigma_B)$ is convex on the compact and convex set $\cS(A)$ by Lemma~\ref{lem:convexity}, and (b) the function $\sigma_B\mapsto f(\rho_A, \sigma_B)$ is concave on $\cS(B)$ by the operator concavity of $\sigma\mapsto\sigma^\lambda$ for $\lambda\in(0,1)$. Thus, Sion's minimax theorem applies. This allows us to exchange the minimization and the maximization in \eqref{eq:additivity-1} and obtain
\begin{align}
I_{\alpha}(\cN)
=&\min_{\sigma_B}\frac{\alpha}{\alpha-1}\log\min_{\rho_A}\left\|\left(\rho_A^\frac{1}{2\alpha}
  \ox\sigma_B^\frac{1-\alpha}{2\alpha}\right)\cN(\Phi_{AA'})
  \left(\rho_A^\frac{1}{2\alpha}\ox\sigma_B^\frac{1-\alpha}{2\alpha}\right)\right\|_\alpha \nb\\
=&\min_{\sigma_B}\frac{\alpha}{\alpha-1}
  \log\left\|\Gamma_{\sigma_B}^{(\alpha)}\circ\cN\right\|_{\cb,1\rar\alpha}, \label{eq:additivity-2}
\end{align}
where $\Gamma_{\sigma_B}^{(\alpha)}(\cdot):=\sigma_B^\frac{1-\alpha}{2\alpha}(\cdot)
\sigma_B^\frac{1-\alpha}{2\alpha}$, and for the last line we have used \eqref{eq:cbquasinorma}. Now we are ready to show the subadditivity. With \eqref{eq:additivity-2}, we have
\begin{align}
     I_{\alpha}(\cN_1\ox\cN_2)
 =  &\min_{\sigma_{B_1B_2}}\frac{\alpha}{\alpha-1}
     \log\left\|\Gamma_{\sigma_{B_1B_2}}^{(\alpha)}\circ(\cN_1\ox\cN_2)\right\|_{\cb,1\rar\alpha} \nb\\
\leq&\min_{\bar{\sigma}_{B_1},\tilde{\sigma}_{B_2}}\frac{\alpha}{\alpha-1}\log
     \left\|\Gamma_{\bar{\sigma}_{B_1}\ox\tilde{\sigma}_{B_2}}^{(\alpha)}\circ(\cN_1\ox\cN_2)\right\|
     _{\cb,1\rar\alpha} \nb\\
 =  &\min_{\bar{\sigma}_{B_1},\tilde{\sigma}_{B_2}}\frac{\alpha}{\alpha-1}\log
     \left\|\left(\Gamma_{\bar{\sigma}_{B_1}}^{(\alpha)}\circ\cN_1\right)
     \ox\left(\Gamma_{\tilde{\sigma}_{B_2}}^{(\alpha)}\circ\cN_2\right)\right\|_{\cb,1\rar\alpha}.
     \label{eq:additivity-3}
\end{align}
Applying the multiplicativity result of Theorem~\ref{thm:multiplicativity} to \eqref{eq:additivity-3}, we eventually get
\begin{equation}\label{eq:subadditivity}
I_{\alpha}(\cN_1\ox\cN_2) \leq I_{\alpha}(\cN_1)+I_{\alpha}(\cN_2).
\end{equation}

It remains to show the superadditivity. By definition, we have
\begin{align}
     I_{\alpha}(\cN_1\ox\cN_2)
 =  &\max_{\rho_{A_1A_2A'_1A'_2}}\min_{\sigma_{B_1B_2}}
     D_\alpha\left(\cN_1\ox\cN_2(\rho_{A_1A_2A'_1A'_2})\big\|\rho_{A_1A_2}\ox\sigma_{B_1B_2}\right) \nb\\
\geq&\max_{\bar{\rho}_{A_1A'_1},\tilde{\rho}_{A_2A'_2}}\min_{\sigma_{B_1B_2}}
     D_\alpha\left(\cN_1(\bar{\rho}_{A_1A'_1})\ox\cN_2(\tilde{\rho}_{A_2A'_2})
     \big\|(\bar{\rho}_{A_1}\ox\tilde{\rho}_{A_2})\ox\sigma_{B_1B_2}\right).
     \label{eq:additivity-4}
\end{align}
Hayashi and Tomamichel~\cite[Lemma 7]{HayashiTomamichel2016correlation} have proved the additivity of quantum R\'enyi mutual information for product states. That is,
\begin{align}
 &\min_{\sigma_{B_1B_2}}
  D_\alpha\left(\cN_1(\bar{\rho}_{A_1A'_1})\ox\cN_2(\tilde{\rho}_{A_2A'_2})
    \big\|(\bar{\rho}_{A_1}\ox\tilde{\rho}_{A_2})\ox\sigma_{B_1B_2}\right) \nb\\
=&\min_{\sigma_{B_1}}D_\alpha\left(\cN_1(\bar{\rho}_{A_1A'_1})\big\|\bar{\rho}_{A_1}\ox\sigma_{B_1}\right)
 +\min_{\sigma_{B_2}}D_\alpha\left(\cN_2(\tilde{\rho}_{A_2A'_2})\big\|\tilde{\rho}_{A_2}\ox\sigma_{B_2}\right).
  \label{eq:additivity-5}
\end{align}
Inserting \eqref{eq:additivity-5} into \eqref{eq:additivity-4} results in
\begin{equation}\label{eq:supadditivity}
I_{\alpha}(\cN_1\ox\cN_2) \geq I_{\alpha}(\cN_1)+I_{\alpha}(\cN_2),
\end{equation}
and we are done.
\end{proof}

\section{Additivity of the Channel Dispersions}
Let $\rho$ and $\sigma$ be two quantum states on the same Hilbert space. the quantum relative entropy of $\rho$ and $\sigma$ is defined as~\cite{Umegaki1954conditional}
\begin{equation}
D(\rho\|\sigma):= \begin{cases}
	\tr(\rho(\log\rho-\log\sigma)) & \text{ if }\supp(\rho)\subseteq\supp(\sigma), \\
	+\infty                        & \text{ otherwise.}
\end{cases}
\end{equation}
Suppose that $\supp(\rho)\subseteq\supp(\sigma)$. The relative entropy variance is defined as~\cite{TomamichelHayashi2013hierarchy,Li2014second} 
\begin{equation}
\label{eq:relative-variance}
V(\rho\|\sigma):=\tr\rho(\log\rho-\log\sigma)^2-(D(\rho\|\sigma))^2.
\end{equation}
For a bipartite quantum state $\rho_{AB}$, the mutual information is given by
\begin{align}
I(A:B)_\rho
:=& D(\rho_{AB} \| \rho_A \ox \rho_B) \nb\\
 =& H(A)_\rho-H(AB)_\rho+H(B)_\rho,           
\end{align}
Where $H(A)_\rho=-\tr(\rho_A\log\rho_A)$ is the von Neumann entropy of the state $\rho_A$. We define the mutual information of a quantum channel $\cN_{A'\rar B}$ as
\begin{equation}\label{eq:mutualI-def}
I(\cN):=\max\left\{I(A:B)_{\cN(\rho_{AA'})}~|~\rho_{AA'}\in\cS(AA')\right\}.
\end{equation}
This quantity quantifies the communication capacity~\cite{BSST1999entanglement,BSST2002entanglement} and simulation cost~\cite{BDHSW2014quantum, BCR2011the} of the channel $\cN_{A'\rar B}$. Before defining the channel dispersions, we let
\begin{equation}\label{eq:capacityachset}
\cO_\cN:=\left\{\id_A\ox\cN(\rho_{AA'})~|~I(A:B)_{\cN(\rho_{AA'})}=I(\cN)\right\}
\end{equation}
denote the set of bipartite states that are generated by the channel $\cN_{A'\rar B}$ and attain the channel mutual information of Eq.~\eqref{eq:mutualI-def}. 

\begin{definition}
The maximal and the minimal dispersions of the channel $\cN_{A'\rar B}$ are defined respectively as
\begin{align}
V_{\rm max}(\cN)&:=\max_{\rho_{AB}\in\cO_\cN}V(\rho_{AB}\|\rho_A\ox\rho_B), \\
V_{\rm min}(\cN)&:=\min_{\rho_{AB}\in\cO_\cN}V(\rho_{AB}\|\rho_A\ox\rho_B).
\end{align}
\end{definition}

\begin{remark}
In the above definition of the channel dispersions, we do not restrict the state $\rho_{AA'}$ in Eq.~\eqref{eq:capacityachset} to be pure. This is in contrast to the definition of the channel mutual information in Eq.~\eqref{eq:mutualI-def}, where without loss of generality we can assume that the state $\rho_{AA'}$ is pure, due to the data processing inequality of the relative entropy. 
\end{remark}

The main result of this section is the following additivity property. We point out that the additivity of $V_{\rm max}$ of Eq.~\eqref{eq:vmaxadd} is already implied by the work~\cite{RTB2023moderate}, and we provide an alternative proof here.
\begin{theorem}\label{thm:additivity2}
Let $\cN_1:\mathcal{L}(A_1')\rar\mathcal{L}(B_1)$ and $\cN_2:\mathcal{L}(A_2')\rar\mathcal{L}(B_2)$ be two quantum channels. We have
\begin{align}
	V_{\rm max}(\cN_1\ox\cN_2)&=V_{\rm max}(\cN_1)+V_{\rm max}(\cN_2), \label{eq:vmaxadd} \\
	V_{\rm min}(\cN_1\ox\cN_2)&=V_{\rm min}(\cN_1)+V_{\rm min}(\cN_2). \label{eq:vminadd}
\end{align}
\end{theorem}

In the remaining part of this section, we first establish two technical lemmas, and then use them to complete the proof of Theorem~\ref{thm:additivity2}.

\begin{lemma}\label{lem:div-radius}
For a quantum channel $\cN_{A'\rar B}$, there exists a unique state $\sigma_B^*$ such that
\begin{equation}\label{eq:div-radius}
I(\cN)=\max_{\rho_{AA'}}D(\cN(\rho_{AA'})\|\rho_A\ox\sigma_B^*).
\end{equation}
Furthermore, if $\rho_{AB}\in\cO_\cN$, then $\rho_B=\sigma_B^*$.
\end{lemma}

\begin{proof}
Let $\rho_{A'}\in\cS(A')$ and $\sigma_B\in\cS(B)$. Let $\rho_{AA'}$ be a purification of $\rho_{A'}$. Define
\begin{equation}
f(\rho_{A'},\sigma_B):=D\left(\id_A\ox\cN(\rho_{AA'})\|\rho_A\ox\sigma_B\right).
\end{equation}
At first, we show that the function $\rho_{A'}\mapsto f(\rho_{A'},\sigma_B)$ is concave. Suppose that $\rho_{A'}=\lambda_1\rho_{A'}^{(1)}+\lambda_2\rho_{A'}^{(2)}$, where $\rho_{A'}^{(1)}$ and $\rho_{A'}^{(2)}$ are density matrices, and $\lambda_1$ and $\lambda_2$ are nonnegative numbers such that $\lambda_1+\lambda_2=1$. Let $\rho_{AA'}^{(1)}$ and $\rho_{AA'}^{(2)}$ be purifications of $\rho_{A'}^{(1)}$ and $\rho_{A'}^{(2)}$, respectively. Denote the vector form of a pure state $\rho$ by $\ket{\rho}$. Then the state $\rho_{\bar{A}AA'}$ given by
\begin{equation}
\ket{\rho_{\bar{A}AA'}}
=\sqrt{\lambda_1}\ket{1}_{\bar{A}}\ket{\rho^{(1)}}_{AA'}+\sqrt{\lambda_2}\ket{2}_{\bar{A}}\ket{\rho^{(2)}}_{AA'}
\end{equation}
is a purification of $\rho_{A'}$. We have
\begin{align}
    &f(\rho_{A'},\sigma_B) \nb\\
  = &D(\id_{\bar{A}A}\ox\cN(\rho_{\bar{A}AA'})\|\rho_{\bar{A}A}\ox\sigma_B) \nb\\
\geq&D(\id_{A}\ox\cN(\rho_{AA'})\|\rho_{A}\ox\sigma_B) \nb\\
  = &\lambda_1D\left(\id_A\ox\cN(\rho_{AA'}^{(1)})\|\rho_A^{(1)}\ox\sigma_B\right)+
     \lambda_2D\left(\id_A\ox\cN(\rho_{AA'}^{(2)})\|\rho_A^{(2)}\ox\sigma_B\right) \nb\\
  = &\lambda_1f(\rho_{A'}^{(1)},\sigma_B) + \lambda_2f(\rho_{A'}^{(2)},\sigma_B),
\end{align}
where the inequality is by the monotonicity of the relative entropy under partial trace~\cite{Lindblad1975completely,Uhlmann1977relative}, and the fourth line is by direct computation. This confirms that the function $\rho_{A'}\mapsto f(\rho_{A'},\sigma_B)$ is indeed concave. It is obviously continuous for any $\sigma_B\in\cS^+(B)$. On the other hand, the function $\sigma_{B}\mapsto f(\rho_{A'},\sigma_B)$ is convex due to the operator concavity of the logarithm, and it is obviously continuous on $\cS^+(B)$ for any $\rho_{A'}\in\cS(A')$. So, we can apply Sion's minimax theorem to deduce that
\begin{align}
I(\cN)&=\max_{\rho_{A'}\in\cS(A')}\inf_{\sigma_B\in\cS^+(B)} f(\rho_{A'},\sigma_B) \nb\\
      &=\inf_{\sigma_B\in\cS^+(B)}\max_{\rho_{A'}\in\cS(A')} f(\rho_{A'},\sigma_B) \nb\\
      &=\min_{\sigma_B\in\cS(B)}\max_{\rho_{AA'}\in\cS(AA')}D(\cN(\rho_{AA'})\|\rho_A\ox\sigma_B). \label{eq:div-radius-1}      
\end{align}
If the minimizer in the last line of Eq.~\eqref{eq:div-radius-1} is unique, then the first part of the statement follows. We will show that this is really true, which also implies the second part of the statement. Let $\sigma^*_B$ be a minimizer of the last line of Eq.~\eqref{eq:div-radius-1} and let $\rho^*_{AB}\in\cO_\cN$. By Eq.~\eqref{eq:div-radius-1} we have
\begin{align}
I(\cN)&=\max_{\rho_{AA'}}D(\cN(\rho_{AA'})\|\rho_A\ox\sigma_B^*) \nb\\
      &\geq D(\rho^*_{AB}\|\rho^*_A\ox\sigma_B^*). \label{eq:div-radius-2} 
\end{align}
On the other hand, by the definition of $I(\cN)$, we also have
\begin{align}
I(\cN)&=D(\rho^*_{AB}\|\rho^*_A\ox\rho_B^*) \nb\\
	  &=D(\rho^*_{AB}\|\rho^*_A\ox\sigma_B^*)-D(\rho_B^*\|\sigma_B^*).
	  \label{eq:div-radius-3} 
\end{align}
Equations~\eqref{eq:div-radius-2} and \eqref{eq:div-radius-3} implies that $D(\rho_B^*\|\sigma_B^*)=0$, which is equivalent to $\rho_B^*=\sigma_B^*$. This completes the proof.
\end{proof}

\begin{lemma}\label{lem:structure}
Let $\rho_{AB_1B_2}=\id_{A}\ox\cN_1\ox\cN_2(\rho_{AA'_1A'_2})$. Suppose
\begin{equation}
I(A:B_1B_2)_{\rho}=I(\cN_1\ox\cN_2).
\end{equation}
Then there is a local isometry $W:A\to\bar{A}_1\bar{A}_2\hat{A}A_1A_2$, such that
\begin{equation}\label{eq:structure}
W\rho_{AB_1B_2}W^\dagger=\sum_{x,y}p_{x,y}\proj{x,y}_{\bar{A}_1\bar{A}_2}\ox\rho^{(x,y)}_{\hat{A}}
                          \ox\rho^{(x)}_{A_1B_1}\ox\rho^{(y)}_{A_2B_2},
\end{equation}
where $\{p_{x,y}\}$ is a probability distribution, and $\{\ket{x,y}\}$ is an orthonormal basis. 
\end{lemma}

\begin{proof}
Let $U_\cN:A'\to BE$ be the Stinespring dilation of the channel $\cN$. For the state $\sigma_{ABE}=U_\cN\sigma_{AA'}U_\cN^\dagger$ we have
\begin{align}
I(A:B)&=H(A)-H(AB)+H(B) \nb\\
      &=H(BE)-H(E)+H(B) \nb\\
      &=H(B|E)+H(B),
\end{align}
where $H(B|E)=H(BE)-H(B)$ is called the conditional entropy. This gives an alternative formula for the channel mutual information:
\begin{equation}\label{eq:cmi-alt}
I(\cN)=\max\big\{H(B|E)_\sigma+H(B)_\sigma~\big|~\sigma_{BE}=U_\cN\sigma_{A'}U_\cN^\dagger\big\}.
\end{equation}

Let $U_{\cN_1}:A_1'\to B_1E_1$ and $U_{\cN_2}:A_2'\to B_2E_2$ be the Stinespring dilations of the channels $\cN_1$ and $\cN_2$, respectively. Then $\rho_{AB_1B_2E_1E_2}=(U_{\cN_1}\ox U_{\cN_2})\rho_{AA_1'A_2'}(U_{\cN_1}\ox U_{\cN_2})^\dagger$ is a purification of $\rho_{AB_1B_2}$. Therefore, as was shown in~\cite{AdamiCerf1997neumann},
\begin{align}\label{eq:cmi-add-1}
I(\cN_1\ox\cN_2)&=H(B_1B_2|E_1E_2)+H(B_1B_2) \nb\\
                &\leq H(B_1|E_1)+H(B_1) + H(B_2|E_2)+H(B_2) \nb\\
                &\leq I(\cN_1) + I(\cN_2),
\end{align}
where the second line is by the subadditivity of the von Neumann entropy and the conditional entropy~\cite{Wilde2013quantum}. This actually implies the additivity $I(\cN_1\ox\cN_2)=I(\cN_1) + I(\cN_2)$, as the ``$\geq$'' direction is obvious. So, the inequalities in Eq.~\eqref{eq:cmi-add-1} should be equalities. In particular, we have
\begin{equation}\label{eq:cmi-add-plus}
H(B_1B_2|E_1E_2)_\rho=H(B_1|E_1)_\rho + H(B_2|E_2)_\rho.
\end{equation}
The strong subadditivity of von Neumann entropy~\cite{LiebRuskai1973proof} states that the conditional mutual information of a tripartite state $\rho_{XYZ}$ is always non-negative:
\begin{equation}\label{eq:cmi-ssa}
	I(X:Z|Y):=H(XY)+H(ZY)-H(XYZ)-H(Z)\geq 0.
\end{equation}
Using conditional mutual information, we can write Eq.~\eqref{eq:cmi-add-plus} in the following equivalent form:
\begin{align}
I(B_1E_1;B_2|E_2)_\rho=0, \label{eq:cmi-add-2} \\
I(B_1;E_2|E_1)_\rho=0. \label{eq:cmi-add-3}
\end{align}

Hayden {\it et al} in~\cite{HJPW2004structure} have derived the structural characterization of tripartite quantum states which satisfy the strong subadditivity of von Neumann entropy with equality; see Theorem~\ref{thm:HJPW-structure} in the appendix. Applying Theorem~\ref{thm:HJPW-structure} twice, we derive from Eqs.~\eqref{eq:cmi-add-2} and \eqref{eq:cmi-add-3} the structure of $\rho_{B_1B_2E_1E_2}$: there are local isometries $W_1:E_1\to \bar{E}_1E_1^LE_1^R$ and $W_2:E_2\to \bar{E}_2E_2^LE_2^R$, such that
\begin{align}
 &(W_1\ox W_2)\rho_{B_1B_2E_1E_2}(W_1\ox W_2)^\dagger \nb \\
=&\sum_{x,y}p_{x,y}\proj{x}_{\bar{E}_1}\ox\proj{y}_{\bar{E}_2}\ox\rho^{(x)}_{B_1E_1^L}\ox\rho^{(y)}_{B_2E_2^R}\ox\rho^{(x,y)}_{E_1^RE_2^L}, \label{eq:structure-p1}
\end{align} 
with $\{p_{x,y}\}$ a probability distribution, and $\{\ket{x}\}$ and $\{\ket{y}\}$ orthonormal bases. Denote this state as $\rho_{B_1B_2\tilde{E}_1\tilde{E}_2}$, with $\tilde{E}_1=\bar{E}_1E_1^LE_1^R$ and $\tilde{E}_2=\bar{E}_2E_2^LE_2^R$. From the right hand side of Eq.~\eqref{eq:structure-p1}, we can purify it as
$\rho_{\bar{A}_1\bar{A}_2\hat{A}A_1A_2B_1B_2\tilde{E}_1\tilde{E}_2}$, with vector form
\begin{align}
 &\ket{\rho}_{\bar{A}_1\bar{A}_2\hat{A}A_1A_2B_1B_2\tilde{E}_1\tilde{E}_2} \nb \\
=&\sum_{x,y}\sqrt{p_{x,y}}\ket{x,y}_{\bar{A}_1\bar{A}_2}\ket{x,y}_{\bar{E}_1\bar{E}_2}\ket{\rho^{(x)}}_{A_1B_1E_1^L}
\ket{\rho^{(y)}}_{A_2B_2E_2^R}\ket{\rho^{(x,y)}}_{\hat{A}E_1^RE_2^L}, \label{eq:structure-p2}
\end{align} 
Noticing that 
\begin{equation}
\rho_{AB_1B_2\tilde{E}_1\tilde{E}_2}
=(W_1\ox W_2)\rho_{AB_1B_2E_1E_2}(W_1\ox W_2)^\dagger
\end{equation}
is also a purification of $\rho_{B_1B_2\tilde{E}_1\tilde{E}_2}$, we conclude that there is a local isometry $W:A\to\bar{A}_1\bar{A}_2\hat{A}A_1A_2$ such that
\begin{equation}\label{eq:structure-p3}
W\rho_{AB_1B_2\tilde{E}_1\tilde{E}_2}W^\dagger
=\rho_{\bar{A}_1\bar{A}_2\hat{A}A_1A_2B_1B_2\tilde{E}_1\tilde{E}_2}.
\end{equation}
At last, tracing out $\tilde{E}_1\tilde{E}_2$ from both sides of Eq.~\eqref{eq:structure-p3} leads to Eq.~\eqref{eq:structure} and we are done.
\end{proof}

\begin{proof}[Proof of Theorem~\ref{thm:additivity2}]
Suppose that $\rho_{A_1B_1}\in\cO_{\cN_1}$ and $\rho_{A_2B_2}\in\cO_{\cN_2}$. Since the channel mutual information is additive (cf.~Lemma~\ref{lem:additivity-cmi}), we have $\rho_{A_1B_1}\ox\rho_{A_2B_2}\in\cO_{\cN_1\ox\cN_2}$. This immediately implies that
\begin{align}
V_{\rm max}(\cN_1\ox\cN_2)&\geq V_{\rm max}(\cN_1)+V_{\rm max}(\cN_2), \label{eq:add2-proof-00}\\
V_{\rm min}(\cN_1\ox\cN_2)&\leq V_{\rm min}(\cN_1)+V_{\rm min}(\cN_2). \label{eq:add2-proof-01}
\end{align}

To prove the other directions, we pick an arbitrary state $\rho_{AB_1B_2}\in\cO_{\cN_1\ox\cN_2}$. By Lemma~\ref{lem:structure}, there is a local isometry $W:A\to\bar{A}_1\bar{A}_2\hat{A}A_1A_2$, such that
\begin{equation}\label{eq:add2-proof-1}
W\rho_{AB_1B_2}W^\dagger
=\sum_{x,y}p_{x,y}\proj{x,y}_{\bar{A}_1\bar{A}_2}\ox\rho^{(x,y)}_{\hat{A}}\ox\rho^{(x)}_{A_1B_1}\ox\rho^{(y)}_{A_2B_2}.
\end{equation}
We make use of Lemma~\ref{lem:div-radius} to derive more information of the state $\rho_{AB_1B_2}$. Let $\sigma_{B_1}^*$ and $\sigma_{B_2}^*$ be the states that satisfy Eq.~\eqref{eq:div-radius}, associated with $\cN_1$ and $\cN_2$, respectively. Then the above-mentioned additivity of the channel mutual information implies that $\sigma_{B_1}^*\ox\sigma_{B_2}^*$ is the state that satisfies Eq.~\eqref{eq:div-radius} associated with $\cN_1\ox\cN_2$. So, Lemma~\ref{lem:div-radius} tells us that \begin{equation}
\rho_{B_1B_2}=\sigma_{B_1}^*\ox\sigma_{B_2}^*.
\end{equation}
Therefore,
\begin{align}
I(\cN_1\ox\cN_2)
&=D\left(\rho_{AB_1B_2}\|\rho_A\ox\sigma_{B_1}^*\ox\sigma_{B_2}^*\right) \nb\\
&=D\left(W\rho_{AB_1B_2}W^\dagger\|W\rho_AW^\dagger\ox\sigma_{B_1}^*\ox\sigma_{B_2}^*\right) \nb\\
&=\sum_{x,y}p_{x,y}\Big(D\left(\rho^{(x)}_{A_1B_1}\|\rho^{(x)}_{A_1}\ox\sigma_{B_1}^*\right)
                        +D\left(\rho^{(y)}_{A_2B_2}\|\rho^{(y)}_{A_2}\ox\sigma_{B_2}^*\right)\Big) \nb\\
&\leq I(\cN_1) + I(\cN_2), 
\label{eq:add2-proof-2}
\end{align}
where for the third line we have used Eq.~\eqref{eq:add2-proof-1}, and the last line is by Lemma~\ref{lem:div-radius}. Eq.~\eqref{eq:add2-proof-2}, together with the additivity of channel mutual information, yields that
\begin{align}
\forall x,~&D\left(\rho^{(x)}_{A_1B_1}\|\rho^{(x)}_{A_1}\ox\sigma_{B_1}^*\right)=I(\cN_1), \label{eq:add2-proof-3}\\
\forall y,~&D\left(\rho^{(y)}_{A_2B_2}\|\rho^{(y)}_{A_2}\ox\sigma_{B_2}^*\right)=I(\cN_2). \label{eq:add2-proof-4}
\end{align} 

Now, we computer the relative entropy variance of the states $\rho_{AB_1B_2}$ and $\rho_A\ox\rho_{B_1B_2}$. We have
\begin{align}
 &V\big(\rho_{AB_1B_2}\|\rho_A\ox\rho_{B_1B_2}\big) \nb\\
=&V\big(W\rho_{AB_1B_2}W^\dagger\|W\rho_AW^\dagger\ox\sigma_{B_1}^*\ox\sigma_{B_2}^*\big) \nb\\
=&\tr\left[W\rho_{AB_1B_2}W^\dagger\left(\log W\rho_{AB_1B_2}W^\dagger-\log W\rho_AW^\dagger\ox\sigma_{B_1}^*\ox\sigma_{B_2}^*\right)^2\right]-I(\cN_1\ox\cN_2)^2.
\end{align}
Employing Lemma~\ref{lem:structure} and Lemma~\ref{lem:additivity-cmi}, we proceed as 
\begin{align}
 &V\big(\rho_{AB_1B_2}\|\rho_A\ox\rho_{B_1B_2}\big) \nb\\
=&\sum_{x,y}p_{x,y}\tr\left[\left(\rho^{(x)}_{A_1B_1}\ox\rho^{(y)}_{A_2B_2}\right)
  \left(\log\rho^{(x)}_{A_1B_1}-\log(\rho^{(x)}_{A_1}\ox\sigma_{B_1}^*)
  +\log\rho^{(y)}_{A_2B_2}-\log(\rho^{(y)}_{A_2}\ox\sigma_{B_2}^*)\right)^2\right] \nb\\
 &-\big(I(\cN_1)+I(\cN_2)\big)^2. \nb\\
=&2\sum_{x,y}p_{x,y}\tr\left[\rho^{(x)}_{A_1B_1}
 \left(\log\rho^{(x)}_{A_1B_1}-\log(\rho^{(x)}_{A_1}\ox\sigma_{B_1}^*)\right)\right]\cdot\tr\left[\rho^{(y)}_{A_2B_2}
 \left(\log\rho^{(y)}_{A_2B_2}-\log(\rho^{(y)}_{A_2}\ox\sigma_{B_2}^*)\right)\right] \nb\\ 
 &+\sum_xp_x\tr\left[\rho^{(x)}_{A_1B_1}\left(\log\rho^{(x)}_{A_1B_1}-\log(\rho^{(x)}_{A_1}\ox\sigma_{B_1}^*)\right)^2\right] \nb\\
 &+\sum_yp_y\tr\left[\rho^{(y)}_{A_2B_2}\left(\log\rho^{(y)}_{A_2B_2}-\log(\rho^{(y)}_{A_2}\ox\sigma_{B_2}^*)\right)^2\right] \nb\\
 &-2I(\cN_1)I(\cN_2)-I(\cN_1)^2-I(\cN_2)^2, \label{eq:add2-proof-5}
\end{align}
where $p_x=\sum_yp_{x,y}$ and $p_y=\sum_xp_{x,y}$. By Eqs.~\eqref{eq:add2-proof-3} and \eqref{eq:add2-proof-4}, we get from Eq.~\eqref{eq:add2-proof-5} that
\begin{align}
 &V\big(\rho_{AB_1B_2}\|\rho_A\ox\rho_{B_1B_2}\big) \nb\\
=&\ \ \sum_xp_x\tr\left[\rho^{(x)}_{A_1B_1}\left(\log\rho^{(x)}_{A_1B_1}-\log(\rho^{(x)}_{A_1}\ox\sigma_{B_1}^*)\right)^2\right]
   -I(\cN_1)^2 \nb\\
 &+\!\! \sum_yp_y\tr\left[\rho^{(y)}_{A_2B_2}\left(\log\rho^{(y)}_{A_2B_2}-\log(\rho^{(y)}_{A_2}\ox\sigma_{B_2}^*)\right)^2\right] 
   -I(\cN_2)^2 \nb\\
=&V\left(\rho_{\bar{A}_1A_1B_1}\|\rho_{\bar{A}_1A_1}\ox\rho_{B_1}\right)+
  V\left(\rho_{\bar{A}_2A_2B_2}\|\rho_{\bar{A}_2A_2}\ox\rho^{B_2}\right), \label{eq:add2-proof-6}
\end{align}
where $\rho_{\bar{A}_1A_1B_1}=\tr_{\bar{A}_2\hat{A}A_2B_2} [W\rho_{AB_1B_2}W^\dagger]\in\cO_{\cN_1}$ and $\rho_{\bar{A}_2A_2B_2}=\tr_{\bar{A}_1\hat{A}A_1B_1} [W\rho_{AB_1B_2}W^\dagger]\in\cO_{\cN_2}$. Since $\rho_{AB_1B_2}\in\cO_{\cN_1\ox\cN_2}$ is picked arbitrarily, we deduce from  Eq.~\eqref{eq:add2-proof-6} that
\begin{align}
V_{\rm max}(\cN_1\ox\cN_2)&\leq V_{\rm max}(\cN_1)+V_{\rm max}(\cN_2), \label{eq:add2-proof-70}\\
V_{\rm min}(\cN_1\ox\cN_2)&\geq V_{\rm min}(\cN_1)+V_{\rm min}(\cN_2). \label{eq:add2-proof-71}
\end{align}

Eventually, the combinations of Eqs.~\eqref{eq:add2-proof-00} and \eqref{eq:add2-proof-70}, and \eqref{eq:add2-proof-01} and \eqref{eq:add2-proof-71} let us complete the proof.
\end{proof}

\section{Outlook}
First, we may define the more general completely bounded quasi-norms $\|\cdot\|_{\cb,p\rar q}$ for $0<p,q\leq1$. Possibly, this can be done by adapting the definition of the completely bounded norms for which $p,q\geq1$. Can we still keep the equivalence of this definition to \eqref{eq:cbquasinorm} when  $p=1$? What is the range of $(p,q)$ such that $\|\cdot\|_{\cb, p\rar q}$ is multiplicative?

\smallskip
Second, are there any potential applications of these completely bounded quasi-norms? For quantum Markov semigroups, hypercontractivity in terms of the completely bounded norms, and reverse hypercontractivity in terms of the ordinary Schatten qusi-norms, have been investigated in the literature (see, e.g.,~\cite{BeigiKing2016hypercontractivity,BDR2020quantum}). We do not know whether it is interesting to study the reverse hypercontractivity in terms of the completely bounded quasi-norms.


\smallskip
Last, it would be nice if we can identify the operational interpretations of the entropic quantities discussed in this paper, namely, $I_\alpha(\cN)$ with $\alpha\in[\frac{1}{2},1)$, as well as $V_{\rm max}$ and $V_{\rm min}$. In~\cite{LiYao2024operational}, an operational interpretation for the sandwiched quantum R\'enyi divergence with $\alpha\in[\frac{1}{2},1)$ was given, in characterizing the strong converse exponent of quantum covering-type problems including information decoupling. Due to the intrinsic connection between information decoupling and channel simulation~\cite{BDHSW2014quantum,BCR2011the}, we guess that $I_\alpha(\cN)$ with $\alpha\in[\frac{1}{2},1)$ may characterize the strong converse exponent of quantum channel simulation. This is supported by the more recent treatments of classical and classical-quantum channels~\cite{LLY2024large,OYB2024exponents}. On the other hand, $V_{\rm max}$ and $V_{\rm min}$ should characterize the second-order asymptotics of channel communication and channel simulation, supported by the established results for classical channels~\cite{Hayashi2009information,PPV2010channel,CRBT2024channel}, classical-quantum channels~\cite{TomamichelTan2015second}, and partial result for the fully quantum channels~\cite{DTW2016on}. The additivity properties obtained in the present paper is a step forward towards establishing these operational characterizations.

\appendix
\section{Auxiliary Results}
\label{appA} \setcounter{equation}{0} 
\global\long\def\theequation{A.\arabic{equation}}

Carlen and Lieb's convexity of the trace functions in the following lemma is an important technical tool for us.
\begin{lemma}[\cite{CarlenLieb2008a}]\label{lem:convexity}
	For any $1\leq p \leq 2$, $q\geq 1$, and any fixed $d\times d$ matrix $Y$, the trace function
	\begin{equation}
		\Upsilon_{p,q}(X)=\tr\left[(Y^*X^pY)^{\frac{q}{p}}\right],
	\end{equation}
	defined on the set of positive semidefinite $d\times d$ matrices, is convex.
\end{lemma}

We also need the following additivity result.

\begin{lemma}[\cite{AdamiCerf1997neumann,BSST2002entanglement}]\label{lem:additivity-cmi}
	Let $\cN_1:\mathcal{L}(A_1')\rar\mathcal{L}(B_1)$ and $\cN_2:\mathcal{L}(A_2')\rar\mathcal{L}(B_2)$ be two quantum channels. We have
	\begin{equation}
		I(\cN_1\ox\cN_2)=I(\cN_1)+I(\cN_2).
	\end{equation}
\end{lemma}

Hayden, Jozsa, Petz and Winter have derived a structural characterization of tripartite quantum state $\rho_{ABC}$ which satisfies strong subadditivity of von Neumann entropy with equality, namely, $I(A:C|B)_\rho=0$.
\begin{theorem}[\cite{HJPW2004structure}]
\label{thm:HJPW-structure}
Let $\rho_{ABC}$ be a tripartite quantum state. $I(A:C|B)_\rho=0$ holds if and only if there is a decomposition of system $B$ as 
\begin{equation}
	B = \bigoplus_{j} {b_j^L} \otimes {b_j^R}
\end{equation}
into a direct sum of tensor products, such that
\begin{equation}
	\rho_{ABC} = \bigoplus_{j} p_j \sigma^{(j)}_{A b_j^L} \otimes \omega^{(j)}_{b_j^R C},
\end{equation}
for some probability distribution $\{p_j\}$, and states $\sigma^{(j)}$ and $\omega^{(j)}$ of $A b_j^L$ and $b_j^R C$, respectively.
\end{theorem}

\begin{lemma}[Sion's minimax theorem~\cite{Sion1958general}]
	\label{lem:minimax}
	Let $X$ be a convex subset of a vector space $V$ and $Y$ be a compact convex set in a topological vector space $W$.
	Let $f : X \times Y \rar \mathbb{R}$ be such that
	\begin{enumerate}[(i)]
		\item $f(x,\cdot)$ is quasi-concave and upper semi-continuous on $Y$ for each $x \in X$, and
		\item $f(\cdot, y)$ is quasi-convex and lower semi-continuous on $X$ for each $y \in Y$.
	\end{enumerate}
	Then, we have
	\begin{equation}
		\label{eq:minimax}
		\inf_{x \in X} \sup_{y \in Y} f(x,y)= \sup_{y \in Y} \inf_{x \in X} f(x,y).
	\end{equation}
\end{lemma}

\bibliography{references}

\begin{thebibliography}{10}

\bibitem{HHHH2009quantum}
Ryszard Horodecki, Pawe{\l} Horodecki, Micha{\l} Horodecki, and Karol
  Horodecki.
\newblock Quantum entanglement.
\newblock {\em Rev. Mod. Phys.}, 81(2):865--942, 2009.

\bibitem{Hastings2009superadditivity}
Matthew~B. Hastings.
\newblock Superadditivity of communication capacity using entangled inputs.
\newblock {\em Nat. Phys.}, 5(4):255--257, 2009.

\bibitem{DJKR2006multiplicativity}
Igor Devetak, Marius Junge, Christoper King, and Mary~Beth Ruskai.
\newblock Multiplicativity of completely bounded $p$-norms implies a new
  additivity result.
\newblock {\em Commun. Math. Phys.}, 266(1):37--63, 2006.

\bibitem{GuptaWilde2015multiplicativity}
Manish~K. Gupta and Mark~M. Wilde.
\newblock Multiplicativity of completely bounded $p$-norms implies a strong
  converse for entanglement-assisted capacity.
\newblock {\em Commun. Math. Phys.}, 334(2):867--887, 2015.

\bibitem{Jencova2006a}
Anna Jen{\v{c}}cov{\'a}.
\newblock A relation between completely bounded norms and conjugate channels.
\newblock {\em Commun. Math. Phys.}, 266:65--70, 2006.

\bibitem{MDSFT2013on}
Martin M{\"u}ller-Lennert, Fr{\'e}d{\'e}ric Dupuis, Oleg Szehr, Serge Fehr, and
  Marco Tomamichel.
\newblock On quantum {R{\'e}nyi} entropies: a new generalization and some
  properties.
\newblock {\em J. Math. Phys.}, 54:122203, 2013.

\bibitem{WWY2014strong}
Mark~M. Wilde, Andreas Winter, and Dong Yang.
\newblock Strong converse for the classical capacity of entanglement-breaking
  and {Hadamard} channels via a sandwiched {R{\'e}nyi} relative entropy.
\newblock {\em Commun. Math. Phys.}, 331(2):593--622, 2014.

\bibitem{FrankLieb2013monotonicity}
Rupert~L. Frank and Elliott~H. Lieb.
\newblock Monotonicity of a relative {R{\'e}nyi} entropy.
\newblock {\em J. Math. Phys.}, 54:122201, 2013.

\bibitem{Beigi2013sandwiched}
Salman Beigi.
\newblock Sandwiched {R{\'e}nyi} divergence satisfies data processing
  inequality.
\newblock {\em J. Math. Phys.}, 54:122202, 2013.

\bibitem{CMW2016strong}
Tom Cooney, Mil{\'a}n Mosonyi, and Mark~M. Wilde.
\newblock Strong converse exponents for a quantum channel discrimination
  problem and quantum-feedback-assisted communication.
\newblock {\em Commun. Math. Phys.}, 344(3):797--829, 2016.

\bibitem{LiYao2022strong}
Ke~Li and Yongsheng Yao.
\newblock Strong converse exponent for entanglement-assisted communication.
\newblock {\em IEEE Trans. Inf. Theory}, 70(7):5017--5029, 2024.

\bibitem{LiYao2021reliable}
Ke~Li and Yongsheng Yao.
\newblock Reliable simulation of quantum channels: the error exponent.
\newblock {\em IEEE Trans. Inf. Theory}, 71(1):518--529, 2025.

\bibitem{HayashiTomamichel2016correlation}
Masahito Hayashi and Marco Tomamichel.
\newblock Correlation detection and an operational interpretation of the
  {R{\'e}nyi} mutual information.
\newblock {\em J. Math. Phys.}, 57:102201, 2016.

\bibitem{Umegaki1954conditional}
Hisaharu Umegaki.
\newblock Conditional expectation in an operator algebra.
\newblock {\em Tohoku Math. J.}, 6(2):177--181, 1954.

\bibitem{TomamichelHayashi2013hierarchy}
Marco Tomamichel and Masahito Hayashi.
\newblock A hierarchy of information quantities for finite block length
  analysis of quantum tasks.
\newblock {\em IEEE Trans. Inf. Theory}, 59(11):7693--7710, 2013.

\bibitem{Li2014second}
Ke~Li.
\newblock Second-order asymptotics for quantum hypothesis testing.
\newblock {\em Ann. Statist.}, 42(1):171--189, 2014.

\bibitem{BSST1999entanglement}
Charles~H. Bennett, Peter~W. Shor, John~A. Smolin, and Ashish~V. Thapliyal.
\newblock Entanglement-assisted classical capacity of noisy quantum channels.
\newblock {\em Phys. Rev. Lett.}, 83:3081--3084, 1999.

\bibitem{BSST2002entanglement}
Charles~H. Bennett, Peter~W. Shor, John~A. Smolin, and Ashish~V. Thapliyal.
\newblock Entanglement-assisted capacity of a quantum channel and the reverse
  {Shannon} theorem.
\newblock {\em IEEE Trans. Inf. Theory}, 48(10):2637--2655, 2002.

\bibitem{BDHSW2014quantum}
Charles~H. Bennett, Igor Devetak, Aram~W. Harrow, Peter~W. Shor, and Andreas
  Winter.
\newblock The quantum reverse {Shannon} theorem and resource tradeoffs for
  simulating quantum channels.
\newblock {\em IEEE Trans. Inf. Theory}, 60(5):2926--2959, 2014.

\bibitem{BCR2011the}
Mario Berta, Matthias Christandl, and Renato Renner.
\newblock The quantum reverse {Shannon} theorem based on one-shot information
  theory.
\newblock {\em Commun. Math. Phys.}, 306(3):579--615, 2011.

\bibitem{RTB2023moderate}
Navneeth Ramakrishnan, Marco Tomamichel, and Mario Berta.
\newblock Moderate deviation expansion for fully quantum tasks.
\newblock {\em IEEE Trans. Inf. Theory}, 69(8):5041--5059, 2023.

\bibitem{Lindblad1975completely}
G{\"o}ran Lindblad.
\newblock Completely positive maps and entropy inequalities.
\newblock {\em Commun. Math. Phys.}, 40(2):147--151, 1975.

\bibitem{Uhlmann1977relative}
Armin Uhlmann.
\newblock Relative entropy and the {Wigner-Yanase-Dyson-Lieb} concavity in an
  interpolation theory.
\newblock {\em Commun. Math. Phys.}, 54:21--32, 1977.

\bibitem{AdamiCerf1997neumann}
Christoph Adami and Nicolas~J. Cerf.
\newblock {von Neumann} capacity of noisy quantum channels.
\newblock {\em Phys. Rev. A}, 56(5):3470, 1997.

\bibitem{Wilde2013quantum}
Mark~M. Wilde.
\newblock {\em Quantum Information Theory}.
\newblock Cambridge University Press, Cambridge, 2013.

\bibitem{LiebRuskai1973proof}
Elliott.~H. Lieb and Mary~Beth Ruskai.
\newblock Proof of the strong subadditivity of quantum-mechanical entropy.
\newblock {\em J. Math. Phys.}, 14:1938--1941, 1973.

\bibitem{HJPW2004structure}
Patrick Hayden, Richard Jozsa, Denes Petz, and Andreas Winter.
\newblock Structure of states which satisfy strong subadditivity of quantum
  entropy with equality.
\newblock {\em Commun. Math. Phys.}, 246:359--374, 2004.

\bibitem{BeigiKing2016hypercontractivity}
Salman Beigi and Christopher King.
\newblock Hypercontractivity and the logarithmic sobolev inequality for the
  completely bounded norm.
\newblock {\em J. Math. Phys.}, 57:015206, 2016.

\bibitem{BDR2020quantum}
Salman Beigi, Nilanjana Datta, and Cambyse Rouz{\'e}.
\newblock Quantum reverse hypercontractivity: its tensorization and application
  to strong converses.
\newblock {\em Commun. Math. Phys.}, 376(2):753--794, 2020.

\bibitem{LiYao2024operational}
Ke~Li and Yongsheng Yao.
\newblock Operational interpretation of the sandwiched {R\'enyi} divergence of
  order 1/2 to 1 as strong converse exponents.
\newblock {\em Commun. Math. Phys.}, 405(2):22, 2024.

\bibitem{LLY2024large}
Shi-Bing Li, Ke~Li, and Lei Yu.
\newblock Large deviation analysis for the reverse shannon theorem.
\newblock {\em arXiv:2410.07984}, 2024.

\bibitem{OYB2024exponents}
Aadil Oufkir, Yongsheng Yao, and Mario Berta.
\newblock Exponents for classical-quantum channel simulation in purified
  distance.
\newblock {\em arXiv:2410.10770}, 2024.

\bibitem{Hayashi2009information}
Masahito Hayashi.
\newblock Information spectrum approach to second-order coding rate in channel
  coding.
\newblock {\em IEEE Trans. Inf. Theory}, 55(11):4947--4966, 2009.

\bibitem{PPV2010channel}
Yury Polyanskiy, H.~Vincent Poor, and Sergio Verd{\'u}.
\newblock Channel coding rate in the finite blocklength regime.
\newblock {\em IEEE Trans. Inf. Theory}, 56(5):2307--2359, 2010.

\bibitem{CRBT2024channel}
Michael~X. Cao, Navneeth Ramakrishnan, Mario Berta, and Marco Tomamichel.
\newblock Channel simulation: Finite blocklengths and broadcast channels.
\newblock {\em IEEE Trans. Inf. Theory}, 70(10):6780--6808, 2024.

\bibitem{TomamichelTan2015second}
Marco Tomamichel and Vincent Y.~F. Tan.
\newblock Second-order asymptotics for the classical capacity of image-additive
  quantum channels.
\newblock {\em Commun. Math. Phys.}, 338:103--137, 2015.

\bibitem{DTW2016on}
Nilanjana Datta, Marco Tomamichel, and Mark~M. Wilde.
\newblock On the second-order asymptotics for entanglement-assisted
  communication.
\newblock {\em Quantum Inf. Process.}, 15(6):2569--2591, 2016.

\bibitem{CarlenLieb2008a}
Eric~A. Carlen and Elliott~H. Lieb.
\newblock A {Minkowski} type trace inequality and strong subadditivity of
  quantum entropy {II}: convexity and concavity.
\newblock {\em Lett. Math. Phys.}, 83(2):107--126, 2008.

\bibitem{Sion1958general}
Maurice Sion.
\newblock On general minimax theorems.
\newblock {\em Pac. J. Math.}, 8(1):171--176, 1958.

\end{thebibliography}
\bibliographystyle{unsrt}

\end{document}